\documentclass{article}
\usepackage{amsthm}
\usepackage{amsmath}
\usepackage{amssymb}
\usepackage{mathabx, mathrsfs,latexsym}

\newtheorem{thm}{Theorem}[section]

\newtheorem{lem}[thm]{Lemma}

\theoremstyle{definition}
\newtheorem{defn}[thm]{Definition}
\newtheorem{exm}[thm]{Example}
\newtheorem{rem}[thm]{Remark}
\theoremstyle{remark}

\begin{document}

\setcounter{page}{1}
\begin{center}
{\LARGE On Some Ternary LCD Codes}\\[.5cm]
{\large Nitin S. Darkunde\footnote{\scriptsize School of Mathematical Sciences, Swami Ramanand Teerth Marathwada University, Nanded, India, Email: darkundenitin@gmail.com}, 
Arunkumar R. Patil\footnote{\scriptsize Department of Mathematics, Shri Guru Gobind Singhji Institute of Engineering and Technology, Nanded, India, Email: arun.iitb@gmail.com}}
\end{center}
\markboth{Nitin Darkunde}{Ternary LCD codes}



\


\begin{abstract}
The main aim of this paper is to study $LCD$ codes. Linear code with complementary dual($LCD$) are those codes which have their intersection with their dual code as $\{0\}$. In this paper we will give rather alternative proof of Massey's theorem\cite{8}, which is one of the most important characterization of $LCD$ codes. Let $LCD[n,k]_3$ denote the maximum of possible values of $d$ among $[n,k,d]$ ternary $LCD$ codes. In \cite{4}, authors have given upper bound on $LCD[n,k]_2$ and extended this result for $LCD[n,k]_q$, for any $q$, where $q$ is some prime power. We will discuss cases when this bound is attained for $q=3$.
\end{abstract}



\section{Introduction}	

A linear code with complementary dual (or $LCD$ code) was first introduced
by Massey\cite{8} in 1964. Afterwards, $LCD$ codes were extensively studied and applied in different fields.
Recently, Dougherty
et al.\cite{4} gave a linear programming bound on the largest size of an $LCD$ code. In 2015, Carlet and Guilley \cite{1} have given different types of  constructions of
$LCD$ codes. Further in 2017, Galvez et al.\cite{4} gave bounds on $LCD$ codes in binary case.

Let $GF(q)$ be a finite field with $q$ elements\cite{6,9}, where $q=p^k$, for some prime $p$ and $k\in \mathbb{Z}_{+}$. By $(GF(q))^n$ , we mean a cartesian product of $GF(q)$ with itself $n$ number of times, which is a vector space of dimension $n$ over $GF(q)$. A $k-$dimensional vector subspace of $(GF(q))^n$ over $GF(q)$ is called as {\em  $[n,k]_q$-linear code}\cite{9}. For a linear code $C$, its (minimum) distance\cite{9} is denoted by $d= d(C)$ and defined as $\min\left\{d(x,y):x\neq y, x,y\in C\right\}$, where $d(x, y)$ is usual Hamming distance between two codewords in $C$.
  These values of $n, k, d$ are called as parameters of corresponding code. 
  A {\em generator matrix}\cite{9}  for a code
  $C$ is denoted by matrix $G$ whose row vectors form a basis for $C$, whereas a {\em parity check matrix}\cite{9}  $H$ for code $C$ is a matrix
  whose rows form a basis for  dual code $C^ \bot$.
  Also,  $v\in C \Longleftrightarrow vH^T=0$ and $v\in C^\perp\Longleftrightarrow vG^T=0$. A linear code of distance $d$ is $u$-{\em error-detecting}\cite{9}
  $\Longleftrightarrow$ $d\geq u+1$, whereas a code $C$ is $v$-{\em error-correcting}\cite{6,9} $\Longleftrightarrow$ $d\geq 2v+1$, where $u, v\in \mathbb{Z}_{+}$. Hence $t=\left \lfloor \frac{(d-1)}{2} \right \rfloor$, is the error correcting capability of a code. For practical purposes we should have linear codes with distance as large as possible.

\section{Preliminaries}
Here, we will see a brief introduction of $LCD$ codes.

\begin{defn} (\cite{4,8}). A linear code with complementary dual is a code $C$, for which we have $C \cap C^{\perp}=\{0\}.$
\end{defn}

\begin{exm}
$C=\{00,01\}\subseteq (GF(2))^2$.\\
\end{exm}
There are some linear codes which are not $LCD$. For example: $C=\{0000, 1010,0101, 1111\}\subseteq (GF(2))^4$ is not $LCD$ code, because for this code, we have $C^{\perp}=\{0000,1010,0101,1111\}$ and hence, their intersection is non trivial.

Note that, if $C$ is $LCD$ code, then so is $C^{\perp}$. Let us state an important Theorem given by Massey in \cite{8} and give its alternate proof, which is new to the best of our knowledge, as we haven't made any use of idea of orthogonal projector, which has been used by Massey.

\begin{thm}$($\cite{8}$).$ \label{Massey's Theorem}
Let $G$ be a generator matrix of a linear code over $GF(q)$. Then $G$ generates an $LCD$ code if and only if $GG^{T}$ is invertible matrix.\\
\end{thm}
\begin{proof}
Suppose $det(GG^{T})\neq 0$. We need to prove that $C$ is an $LCD$ code. Suppose $C$ is not $LCD$ code. Therefore there exists a non zero vector $v\in C \cap C^{\perp}$. Hence, we get $v \in C$ and $v \in C^{\perp}$. Since $v \in C $, therefore $\exists \hspace{0.03in} u \neq 0$ in $(GF(q))^k$ such that $v=uG$, where $G$ is given to be a generator matrix for $C$. Next $v \in C^{\perp}$, as a result of which, we get that $vG^{T}=0$. Consequently, $uGG^{T}=0$. Call $GG^{T}$ as $A$. But by hypothesis $A\in GL(k,GF(q))$. Hence we get homogeneous system $uA=0$, post-multiplying both sides by $A^{-1}$, we get $u=0$ and therefore we have,  $v=0$, which is a contradiction to the hypothesis. Therefore, whenever $GG^{T}$ is invertible, then linear code generated by $G$ must be $LCD$ code.

Conversely, suppose $C$ is $LCD$ code. We need to prove that $det(GG^{T})\neq 0$. Suppose $det(GG^{T})=0$. Therefore $GG^{T}$ is a singular linear transformation, hence there exists non zero vector $u \in (GF(q))^k$ such that $uGG^{T}=0$. Let $v=uG$, which implies $v\neq 0$ and we get $vG^{T}=0$, hence $v \in C^{\perp}$. Now it remains to show that $v \in C$. Since we had taken $v$ to be a non zero vector in $(GF(q))^n$ such that $v=uG$, we get $v\in C$. Therefore $\exists \hspace{0.03in} v \neq 0$ in $C \cap C^{\perp}$. 
\end{proof}

\section{Elementary bounds}
In this section, we are only concerned with codes over ternary field. Dougherty et al.\cite{3} introduced a concept of $LCD[n,k]$ over binary fields. Recently Galvez et al.\cite{4} had given an upper bound on $LCD[n,k]$ in binary case and also given some exact values for $k=2$ and for any $n$. They also extended this result for arbitrary values of $q$. Here we will obtain exact values of $LCD[n,k]$ in ternary case. Determination of values of $LCD[n,k]$ is analogous to determination of $A_q(n,d)$, where in the former case we used to concentrate on $d$ and in a later case we used to concentrate on size of a code. Firstly, let us have some definitions.

\begin{defn} For fixed values of $n$ and $k$, we have 
\begin{enumerate}
\item $LCD[n,k]:=\text{max}\{d:$ there exists a binary $[n,k,d]$\hspace{0.03in} $LCD$ code$\}.$
\item $LCD[n,k]_3:=\text{max}\{d:$ there exists a ternary $[n,k,d]$\hspace{0.03in} $LCD$ code$\}.$
\end{enumerate}
\end{defn}
Now we state a remark , which was a consequence of Lemma 2 from \cite{4}.

\begin{rem}
$LCD[n,k]_q \leq \left \lfloor \frac{n.q^{k-1}}{q^{k}-1} \right \rfloor$, for $k \geq 1$.\\
As a consequence of it, for $q=3$ and $k=2$, we have $LCD[n,2]_3 \leq \left \lfloor \frac{3n}{8} \right \rfloor$. 

\end{rem}

Now based on bound given above, we can obtain exact values of $LCD[n,2]_3$.

\begin{thm}
Let $n\geq 2$. Then $LCD[n,2]_3 = \left \lfloor \frac{3n}{8} \right \rfloor$, for $n \equiv 3, 4 (\text{mod}\hspace{0.1in}9)$.
\end{thm}

\begin{proof}
Our aim is to show the existence of $LCD$ codes with minimum distance achieving the bound in above remark.
\begin{enumerate}
\item Let $n \equiv 3(\text{mod}\hspace{0.1in}9)$, i.e. $n=9m+3$, for some $m\in \mathbb{Z}_{+}$. Consider the linear code with the following generator matrix.\\

$$G=
\left[\begin{array}{c|c|c}
1\ldots 1 & 2\ldots 2 & 0\ldots 0\\
\underbrace{0\ldots 0}_{3m} & \underbrace{0\ldots 0}_{3m+2} & \underbrace{2\ldots 2}_{3m+1}
\end{array}\right].
$$

This code has minimum weight $3m+1=\left \lfloor \frac{3(9m+3)}{8} \right \rfloor $ and 
$GG^{T}=\begin{bmatrix}
1 & 0\\
0 & 2
\end{bmatrix}$. Hence $det(GG^{T})=2\not\equiv 0 (\text{mod}\hspace{0.1in}3)$ and therefore this matrix is invertible. By Theorem \ref{Massey's Theorem} above, this code is an $LCD$ code.

\item Let $n \equiv 4(\text{mod}\hspace{0.1in}9)$, i.e. $n=9m+4$, for some $m\in \mathbb{Z}_{+}$. Consider the linear code with the following generator matrix.\\
$$G=
\left[\begin{array}{c|c|c}
1\ldots 1 & 2\ldots 2 & 0\ldots 0\\
\underbrace{0\ldots 0}_{3m+1} & \underbrace{0\ldots 0}_{3m+2} & \underbrace{2\ldots 2}_{3m+1}
\end{array}\right].
$$

This code has minimum weight $3m+1=\left \lfloor \frac{3(9m+4)}{8} \right \rfloor $ and 
$GG^{T}=\begin{bmatrix}
2 & 0\\
0 & 2
\end{bmatrix}$. Hence $det(GG^{T})=4\not\equiv 0 (\text{mod}\hspace{0.1in}3)$ and therefore this matrix is invertible. By Theorem \ref{Massey's Theorem} above, this code is an $LCD$ code.
\end{enumerate}
\end{proof}

Now we will give one construction of ternary $LCD$ codes from primary constructions of linear codes. As far as we know, this construction have not yet been studied in the literature of $LCD$ codes.

\begin{defn}$($\cite{9}$)$.
Let $q$ be odd. Let $C_i$ be an $[n,k_i,d_i]$ linear code over $GF(q)$, for $i=1, 2$. Define $C_1\between C_2:= \{(c_1+c_2, c_1-c_2): c_1\in C_1, c_2\in C_2\}$. Then $C_1\between C_2$ is a linear code over $GF(q)$. This code is $[2n, k_1+k_2]$-linear code over $GF(q)$.
\end{defn}

\begin{rem}
If $G_1$ and $G_2$ is generator matrix of $C_1$ and $C_2$ respectively, then generator matrix $G$ of $C_1\between C_2$ is given by 
$G=\begin{bmatrix}
G_1 & G_1\\
G_2 & -G_2
\end{bmatrix}$. 
\end{rem}

\begin{thm}
Let $C_i$ be $[n, k_i]$ $LCD$ codes over $GF(3)$, for $i=1, 2$. Then $C_1\between C_2$ is also a $LCD$ code over $GF(3)$.
\end{thm}

\begin{proof}
It is given that $C_1$ and $C_2$ both are $LCD$ codes over $GF(3)$. Suppose $G_1$ is generator matrix of $C_1$ and $G_2$ is generator matrix of $C_2$. Therefore by Theorem $2.3$ above, we have $det(G_1G_1^T)\not\equiv 0 (\text{mod}\hspace{0.1in}3)$ and $det(G_2G_2^T)\not\equiv 0 (\text{mod}\hspace{0.1in}3)$. Therefore, we have $GG^T=\begin{bmatrix}
G_1 & G_1\\
G_2 & -G_2

\end{bmatrix}\begin{bmatrix}
G_1^T & G_2^T\\
G_1^T & -G_2^T

\end{bmatrix}$. As a result of it, we get  $GG^T=\begin{bmatrix}
2G_1G_1^T & 0\\
0 & 2G_2G_2^T

\end{bmatrix}$. Now it remains to show that matrix $GG^T$ is invertible. Here $det(GG^T)=det(2G_1G_1^T). det(2G_2G_2^T)= 2^{k_1}det(G_1G_1^T). 2^{k_2} det(G_2G_2^T)= 2^{k_1+k_2}. det(G_1G_1^T). det(G_2G_2^T)$. In this expression both the terms at the end are not divisible by $3$ and $3 \nmid 2^{k_1+k_2}$. Therefore by Euclid's lemma, we get $3 \nmid 2^{k_1+k_2}. det(G_1G_1^T). det(G_2G_2^T)$ and consequently $C_1\between C_2$ is ternary $LCD$ code.
\end{proof}

\begin{lem}
For $n$ and $k$ integers greater than $0$, $LCD[n+1,k]_3\geq LCD[n,k]_3$.
\end{lem}
\begin{proof}
Proof follows on similar lines as that of Lemma $3.1$ from \cite{3}.

\end{proof}

\begin{thm} $(i)$ If $n$ is an integer such that $3\nmid n$, then $LCD[n,1]_3=n$ and $LCD[n,n-1]_3=2$.\\
$(ii)$ If $n$ is an integer such that $3 \nmid (n-1)$, then $LCD[n,1]_3=n-1$ and $LCD[n,n-1]_3=2$.

\end{thm}
\begin{proof}
$(i)$ Consider ternary repetition code $C=\{\underbrace{0\ldots 0}_{n}, \underbrace{1\ldots 1}_{n}, \underbrace{2\ldots 2}_{n}\}$. This code is $[n,1,n]_3$ code, which have largest possible minimum distance. There are two choices for its generator matrices say $G_1$ and $G_2$. Suppose $G_1=\begin{bmatrix}1 & 1 & \ldots & 1 \end{bmatrix}$ and $G_2=\begin{bmatrix}2 & 2 & \ldots & 2 \end{bmatrix}$ respectively. Then $det(G_1G_1^T)=n$ and $det(G_2G_2^T)=2^2n$. Since, $3\nmid n$, we have $det(G_1G_1^T) \not\equiv 0 (\text{mod}\hspace{0.1in}3)$ and $det(G_2G_2^T) \not\equiv 0 (\text{mod}\hspace{0.1in}3)$. Hence by Theorem $2.3$ above, rows of these generator matrices will generate $LCD$ codes. Thus we get, $LCD[n,1]_3=n$. Also, we know that if $C$ is $LCD$ then so its dual $C^\perp$. In this case dual code is $LCD$ code having $dimension$ as $n-1$. If $(c_1,c_2,\ldots,c_n)\in C^\perp$, then $c_1+\cdots+c_n \equiv 0 (\text{mod}\hspace{0.1in}3)$ and hence we will have a choice of codeword $(1,2,0,\ldots,0)$, whose 
weight is minimum. Therefore,
 we get $LCD[n,n-1]_3=2.$\\
$(ii)$ If $3\mid n$, then ternary repetition code $C$ of length $n$ having generator matrix 
$G=\begin{bmatrix}
1\ldots 1 
\end{bmatrix}$
will not be a $LCD$ code, since in this case, $det(GG^T)=n$. So we must try for another ternary code $\displaystyle{\widetilde{C}}$ having a basis as $\mathcal{B}=\{0\underbrace{1\ldots1}_{n-1}\}$. Then we get $\displaystyle{\widetilde{C}}=\{0\underbrace{0\ldots 0}_{n-1}, 0\underbrace{1\ldots1}_{n-1}, 0\underbrace{2\ldots2}_{n-1} \}$. Note that, this code $\displaystyle{\widetilde{C}}$ have maximum possible minimum distance amongst all ternary linear codes, besides ternary repetition code. In present case, there are two choices for its generator matrices, say $G_1=
\begin{bmatrix}
0 & \underbrace{1\ldots1}_{n-1}
\end{bmatrix}$ 
and 
$G_2=\begin{bmatrix}
0 & \underbrace{2\ldots2}_{n-1}
\end{bmatrix}$. 
As a result of which, we get $G_1G_1^T=n-1$ and $G_2G_2^T=2^2.(n-1)$. Consequently, $det(G_1G_1^T)=n-1$ and $det(G_2G_2^T)=2^2.(n-1)$. Hence by Theorem $2.3$ above, $G_1$ and $G_2$ will generate ternary $LCD$ code $\displaystyle{\widetilde{C}}$ if and only if $3\nmid (n-1)$. 

Further, we know that if $\displaystyle{\widetilde{C}}$ is $LCD$ then so its dual $\displaystyle{{\widetilde{C}}^\perp}$. In this case, dual code is $LCD$ code having $dimension$ as $n-1$. If $(c_1,c_2,\ldots,c_n)\in \displaystyle{\widetilde{C}}^\perp$, then $c_2+\cdots+c_n \equiv 0 (\text{mod}\hspace{0.1in}3)$ and hence we will have a choice of codeword $(0,0,\ldots,1,2)$ whose weight is minimum. Therefore, we get $LCD[n,n-1]_3=2.$\\
\end{proof}
\section{Conclusion}
In this paper, We have given new construction of ternary $LCD$ codes, by using some primary constructions. Also, we have discussed some cases where the bound on $LCD[n,k]_3$ is attained.  In a future study, we will generalize this result for any $q$.



\end{document}